\documentclass[letterpaper, 11pt]{article}
\usepackage{appendix}
 
\let\oldenumerate\enumerate
\renewcommand{\enumerate}{
  \oldenumerate
  \setlength{\itemsep}{0pt}
  \setlength{\parskip}{0pt}
  \setlength{\parsep}{0pt}
}

\newcommand{\comment}[1]{}
 
\usepackage{graphicx}

\usepackage{subfigure}
\usepackage{amsmath}
\usepackage{amssymb}
\usepackage{mdwlist}

\usepackage{hyperref}

\usepackage{xspace}
\usepackage{setspace}

\newcommand{\sv}{{\mathcal V}}
\newcommand{\se}{{\mathcal E}}

\newenvironment{proof}{\noindent {\bf Proof:}~}{\hspace*{\fill}\(\Box\)}
\newenvironment{proofSketch}{\noindent{\bf Proof Sketch:}}{\hspace*{\fill}\(\Box\)}
\newtheorem{theorem}{Theorem}

\newtheorem{claim}{Claim}

\newtheorem{definition}{Definition}

\newtheorem{lemma}{Lemma}

\def\noflash#1{\setbox0=\hbox{#1}\hbox to 1\wd0{\hfill}}

\newcommand{\scripte}{\mathcal{E}}
\newcommand{\scriptv}{\mathcal{V}}

\newcommand{\Zightarrow}{\rightarrow}

\begin{document}
\title{Crash-Tolerant Consensus in Directed Graphs\footnote{\normalsize This research is supported in part by National Science Foundation awards 1329681. Any opinions, findings, and conclusions or recommendations expressed here are those of the authors and do not necessarily reflect the views of the funding agencies or the U.S. government.}}

\author{Lewis Tseng$^{1}$ and Nitin Vaidya$^{2}$\\~\\
 \normalsize $^1$ Department of Computer Science,\\
 \normalsize $^2$ Department of Electrical and Computer Engineering, 
\\ \normalsize University of Illinois at Urbana-Champaign\\~\\ \normalsize Email: \{ltseng3, 
nhv\}@illinois.edu \\ \normalsize  Phone: +1 217-244-6024, +1 217-265-5414
\\ \normalsize Mailing address: Coordinated Science Lab., 1308 West Main St., Urbana, IL 61801, U.S.A.} 

\date{December 30, 2014\footnote{Revised January 1, 2014 to make minor improvements to the presentation.}}
\maketitle



~

~

\begin{abstract}
\normalsize

This work considers a point-to-point network of $n$ nodes connected by {\em directed} links, and proves {\em tight} necessary and sufficient conditions on the underlying communication graphs for
achieving consensus among these nodes under {\em crash} faults. We identify the conditions in both synchronous and asynchronous systems.



\end{abstract}

~

~


~


\thispagestyle{empty}
\newpage
\setcounter{page}{1}

\section{Introduction}
\label{s_intro}

In this work, we explore algorithms for achieving consensus in the presence of crash faults \cite{AA_nancy,welch_book}. We assume a point-to-point network, which is modeled as a {\em directed} graph, i.e., the communication links between neighboring nodes are not necessarily bi-directional. We consider both synchronous and asynchronous systems. 

The crash consensus problem \cite{AA_nancy,welch_book} considers $n$ nodes, of which at most $f$ nodes may crash. The faulty nodes may fail stop at any point of time. We do not assume Byzantine behavior \cite{psl_BG_1982} in this work. A crash consensus algorithm is {\em correct} if it satisfies the following three properties:

\begin{itemize}
\item \textbf{Agreement}: the output (i.e., decision) at all the fault-free nodes is identical.

\item \textbf{Validity}: the output at any fault-free node must be some node's input. 

\item \textbf{Termination}: every fault-free node eventually decides on an output.

\end{itemize}

This paper presents {\em tight} necessary and sufficient conditions for crash consensus in {\em directed} graphs.

\subsection{Related Work}

Lamport, Shostak, and Pease introduced the Byzantine consensus problem in \cite{lamport_agreement}. Subsequently, researchers also explored the consensus problem in the presence of crash faults \cite{welch_book,AA_nancy}. It has been shown that the lower bound on the round complexity is $f+1$, and $f+1$ nodes are sufficient for solving crash consensus \cite{welch_book,AA_nancy}. For undirected graphs, it is easy to see that $f+1$ node connectivity is both necessary and sufficient for crash consensus.

For Byzantine consensus in undirected graphs, \cite{impossible_proof_lynch, dolev_82_BG} showed that $2f+1$ node connectivity is both necessary and sufficient. Recently, we identified tight conditions for Byzantine consensus problem in {\em directed} graphs \cite{Tseng_Exact}.
For link failures in complete graphs, Schmid, Weiss, and Keidar proved impossibility results and lower bound on the number of nodes for synchronous consensus under {\em transient Byzantine link} faults \cite{Schmid_link}; however, the nodes are always fault-free. Many effort has also been devoted to characterizing tight conditions for other related problems. Please refer to our prior work \cite{Tseng_Exact} for more details.

For crash faults, Charron-Bost et al. proved tight conditions for {\em approximate} consensus in dynamic graphs \cite{approximate_consensus_dynamic}, where the graphs may change continually and unpredictably, in synchronous and partially-synchronous systems. Our work considers {\em exact} and {\em approximate} consensus in synchronous and asynchronous systems, respectively. Moreover, we assume the communication graph is {\em static}.

\subsection{Network Model}

Sections \ref{s:sync} and \ref{s:iterative} assume synchronous systems, and section \ref{s:async} considers asynchronous systems. The underlying communication network is {\em static}, i.e., it does not change over time.
The communication network consisting
of $n$ nodes is modeled as a simple {\em directed} graph $G(\scriptv,\scripte)$, where $\scriptv$ is the set of $n$ nodes, and $\scripte$ is the set of directed edges between the nodes in $\scriptv$. We assume that $n\geq 2$, since the consensus problem for $n=1$ is trivial.  
Node $i$ can transmit messages to another node $j$ if and only if the directed edge $(i,j)$ is in $\scripte$. Also, each node can send messages to itself as well. 
For node $i$, let $N_i^-$ be the set of nodes from which $i$ can receive messages. 
That is, $N_i^- = \{\, j ~|~ (j,i)\in \scripte\, \} \cup \{i\}$.
Define $N_i^+$ as the set of nodes that can receive messages from node $i$.
That is, $N_i^+ = \{\, j ~|~ (i,j)\in \scripte\, \} \cup \{i\}$. 

All the communication links are reliable,
FIFO (first-in first-out) and deliver each transmitted message exactly once.

\section{Synchronous Systems}
\label{s:sync}

\subsection{Necessary Condition}
\label{s:sync_nec}

All the paths we discuss in the paper are directed paths. We first introduce some useful definitions. A reduced graph $G_F$ for a graph $G(\sv, \se)$ is a subgraph induced by vertex subset $\sv - F$ where $F$ is a potential fault set. The formal definition is presented below.

\begin{definition} {\bf (Reduced Graphs)}
For a given graph $G(\sv, \se)$, and a given parameter $k$, and each set $F \subset V$ such that $|F| \leq k$, reduced graph $G_{F}(\sv_{F}, \se_{F})$ is defined as follows: (i) 
$\sv_{F} = \sv - F$, and (ii) $\se_{F}$ is obtained by removing from $\se$ all the links incident on the nodes in $F$. That is, $\se_{F} = \se - \{(i, j) \in \se~|~i \in F \text{or}~~j \in F\}$.
\end{definition}

We define a fault-tolerant version of node connectivity over a directed graph, which extends the traditional notion of node connectivity (or vertex connectivity) \cite{Graph_theory_west}.

\comment{+++++++
\begin{definition} {\bf (Node Connectivity)} 
\label{def:node}
A graph $G(\sv, \se)$ is said to satisfy $k$ node connectivity if $G$ has more than $k$ nodes, and for any $F \subset \sv$ with $|F| \leq k$, $G_{F}$ remains (weakly) connected.
\end{definition}
+++++++}

\begin{definition} {\bf (Crash-Tolerant Node Connectivity)}
A graph $G(\sv, \se)$ is said to satisfy $k$ Crash-Tolerant Node Connectivity (CT node connectivity) if for any $F \subset \sv$ such that $|F| \leq k$, there is a single node $s \in \sv - F$ that has paths to all the nodes in $G_{F}$. 
\end{definition}
Recall that by assumption, we assume that $i \in N_i^+$ and $N_i^-$, and hence, $i$ has a path to itself as well. The traditional notion of node connectivity \cite{Graph_theory_west}, some reduced graph may not have a node that can reach all the nodes, since it only requires the reduced graph to be weakly connected. 


\begin{definition} {\bf (Directed Rooted Spanning Tree)}
\label{def:tree}
A spanning tree of a graph $H(\sv, \se)$ is said to be a directed rooted spanning tree if there is a single root in the spanning tree that has directed paths to all the nodes in $\sv$.
\end{definition}

It should be easy to see that $k$ CT node connectivity is equivalent to the condition that \textit{given any reduced graph $G_F$, there exists a directed rooted spanning tree}. Charron-bost et al. \cite{approximate_consensus_dynamic} also use the notion of rooted spanning tree to specify the tight condition for achieving approximate consensus in dynamic networks.

\begin{definition} {\bf (Source)}
\label{def:source}
Given a reduced graph $G_F$, a node $s$ is said to be the {\em source} of $G_F$ if there exists a directed rooted spanning tree with $s$ being the root.
\end{definition}
With a slight abuse of terminology, we will use the terms root and source interchangeably. 

~

Based on CT node connectivity, the following theorem presents the necessary condition.

\comment{+++++++++
\begin{theorem}
\label{thm:nc}
If exact consensus is possible in $G(\scriptv,\scripte)$ with at most $f$ crash faults, then for all $F \subset \scriptv$ and $|F| \leq f$, the reduced graph $G_F$ contains at least one directed rooted spanning tree.
\end{theorem}
+++++++++++}

\begin{theorem}
\label{thm:nc}
If exact consensus is possible in $G(\scriptv,\scripte)$ with at most $f$ crash faults, then $G(\sv, \se)$ satisfies $f$ CT node connectivity.
\end{theorem}

\begin{proof}
The proof is by contradiction. Suppose that there exists a consensus algorithm, and $G(\sv, \se)$ does not satisfy $f$ CT node connectivity. Thus, there exists a set $F \subset \sv$ with $|F| \leq f$, and a pair of nodes $i, j \not\in F$ such that there is no node $s$ that has paths from $s$ to both $i$ and $j$. Note that by assumption, each node $i$ has a path to itself. 

For the reduced graph $G_F$ and a node $x$ in $\sv - F$, define $S_x$ as the set of all nodes that have paths to node $x$ in $G_F$. Note that $S_x$ contains $x$ as well, because $x$ has a path to itself. 
By assumption, $S_i$ and $S_j$ are disjoint. Moreover, there is no path from any node in $S_i$ to any node in $S_j$ in $G_{F}$, and vice versa, since otherwise, there exists some node that can reach both nodes $i$ and $j$, which contradicts with the assumption. Then, $V$ can be partitioned into disjoint sets $F, S_i, S_j, R$, where $F, S_i$ and $S_j$ are defined as above, and $R$ contains the remaining nodes, i.e., $R = \sv - F - S_i - S_j$. Then, we make the following observations:

\begin{itemize}
\item $F$ and $R$ may be empty, but $S_i$ and $S_j$ are non-empty, since $i \in S_i$ and $j \in S_j$.

\item Nodes in $R$ (if non-empty) have no path to nodes in $S_i \cup S_j$ in $G_F$ by definition. This is because if some node $r \in R$ can reach some node in $S_i$ or $S_j$ in $G_F$, then by definition, $r$ should also be in $S_i$ or $S_j$, respectively. This contradicts with the assumption. 
\end{itemize}

Now, consider an execution of the consensus algorithm where $F$ (if non-empty) are the faulty nodes which crash before the start of the algorithm. All the other nodes are assumed to be fault-free. This is possible, since by assumption, $|F| \leq f$. Also, suppose that nodes in $S_i$ and nodes in $S_j$ have distinct input values. Without loss of generality, assume that nodes in $S_i$ have input $0$ and nodes in $S_j$ have input $1$. Nodes in $R$ have input either $0$ or $1$. 

Consider a node $x$ in $S_i$. Since in $G_F$, there is no path from $S_j \cup R$ to nodes in $S_i$, the only input value learned by $x$ throughout the execution of the algorithm is $0$, and to satisfy validity property, $0$ should be the output of $x$. Similarly, a node $y$ in $S_j$ can only learn $1$ throughout the execution of the algorithm, and thus, $1$ should also be the output of $y$. Note that by assumption, both $S_i$ and $S_j$ are non-empty, and fault-free. Therefore, the fact that $S_i$ and $S_j$ agree on different outputs violates the agreement property of the algorithm, a contradiction.
\end{proof}

\subsection{Sufficiency}
\label{s:sync_suff}

In this section, we propose a consensus algorithm in graphs that satisfy $f$ CT node connectivity. This section assumes that each node has a binary input. In Section \ref{s:mvc}, we discuss how to extend the algorithm to solve multi-valued consensus. Note that the existence of such correct consensus algorithm proves the following theorem.

\comment{++++++++
\begin{theorem}
\label{thm:sc}
Given $G(\sv, \se)$, if for all $F \subset \sv$ and $|F| \leq f$, $G_F$ contains a directed rooted spanning tree, then consensus is achievable in $G$ with at most $f$ crash faults.
\end{theorem}
++++++++++}

\begin{theorem}
\label{thm:sc}
If $G(\sv, \se)$ satisfies $f$ CT node connectivity, then binary consensus is achievable in $G$ with at most $f$ crash faults.
\end{theorem}

This theorem also implies that $f$ CT node connectivity is a {\em tight} condition for binary consensus. Section \ref{s:mvc} shows that $f$ CT node connectivity is sufficient for multi-valued consensus, as well. Therefore, $f$ CT node connectivity is a {\em tight} condition for consensus in the presence of $f$ crash faults.

\paragraph{Algorithm Min-Max}

For a graph $H$ that contains a directed rooted spanning tree, define $height(r,H)$ as the minimum height of all the spanning trees rooted at $r$ in $H$. That is, 

\[
height(r,H) = \min_{\text{all spanning tree}~T~\text{rooted at}~r~\text{in}~H} \text{height of}~T
\] 

Given a graph $G$, define the fault-tolerant diameter $d$ as follows:

\begin{equation}
\label{eq:d}
d := \max_{F \subset \sv,~~|F| \leq f} ~~\max_{\text{all roots}~s~\text{of}~G_{F}} height(s, G_F)
\end{equation}

Due to the notion of directed rooted spanning tree (Definition \ref{def:tree}), given any reduced graph $G_F$, if no node in $\sv - F$ crashes, then the source of $G_F$ (Definition \ref{def:source}) is able to propagate a value to any other node in $\sv-F$ within $d$ rounds by performing flooding, i.e., a source broadcasts its value in the first round, and then in later rounds, all the nodes forward the value received in the current round.

Now, we present the code running at each node $i$.

~

\hrule

\vspace*{2pt}

\noindent {\bf Algorithm Min-Max}

\vspace*{4pt}

\hrule

\vspace*{4pt}

  \begin{itemize}
  \item Set $v_i$ to the input at node $i$.
  
  \item For Phase $p = 1$ to $2f+2$:

        ~~If $p~mod~2 = 0$, then repeat the following steps $d$ times ({\bf Min Phase}):

            \begin{enumerate}
            \item Broadcast $v_i$ to nodes in $N_i^+$.
            \item Receive the broadcast values from $N_i^-$.
            \item Set $v_i$ to the {\bf minimum} value of all the values received.
            \end{enumerate}

        ~~Else, repeat the following steps $d$ times ({\bf Max Phase}:):

            \begin{enumerate}
            \item Broadcast $v_i$ to nodes in $N_i^+$.
            \item Receive the broadcast values from $N_i^-$. 
            \item Set $v_i$ to the {\bf maximum} value of all the values received.
            \end{enumerate}

   \item Output $v_i$.
   \end{itemize}

\hrule
Note that by definition of, $i \in N_i^-$ and $N_i^+$, so in step 2 of each phase, $i$ can receive the message from itself.

\begin{theorem}
Algorithm Min-Max is correct for binary inputs in all the graphs that satisfy $f$ CT node connectivity.
\end{theorem}

\begin{proof}
Validity and termination properties are obvious, since $d$ is upper bounded by $n$. Now, we prove that the agreement property also holds assuming that the inputs are either $0$ or $1$.

Fix an execution of the algorithm. Since there are $2f+2$ phases. There must exists a pair of consecutive phases $p_t, p_{t+1}$ such that no node crashes in phases $p_t$ and $p_{t+1}$. Without loss of generality, let $p_t$ be the Min Phase and $p_{t+1}$ be the Max Phase.

Denote by $F$ the nodes that have crashed in the execution by the end of Phase $p_{t-1}$. Recall that the source of a reduced graph $H$ is defined as the root of the directed spanning tree in $H$ as per Definition \ref{def:source}. Consider two cases:

\begin{itemize}
\item Case I: if some source $s$ of the reduced graph $G_F$ has $v_s = 0$ at the beginning of phase $p_t$, then due to the definition of the source and fault-tolerant diameter $d$, by the end of phase $p_t$, every node $i \in \sv - F$ has received $0$ on a path from source $s$ and sets $v_i = 0$, since $p_t$ is a Min Phase.

\item Case II: if each source $s$ of the reduced graph $G_F$ has $v_s = 1$ at the beginning of phase $p_t$, then by the end of $p_t$, each source $s$ still has $v_s = 1$. Suppose by way of contradiction that each source $s$ of $G_F$ has $v_s = 1$ at the beginning of phase $p_t$, but by the end of $p_t$, some source $s'$ has $v_{s'} = 0$. By assumption, source $s'$ must receive $0$ on a path from some other {\em non-source} node $x$ in phase $p_t$. However, the fact that there exists a path from $x$ to $s'$ implies that $x$ is also a source in $G_F$, and $v_x = 0$ at the start of phase $p_t$. This is a contradiction. Now, observe that by the end of phase $p_t$, each source $s$ still has $v_s = 1$, and phase $p_{t+1}$ is the Max Phase. Therefore, by the end of $p_{t+1}$, every node $i \in \sv - F$ will receive $1$ on a path from source $s$ and sets $v_i = 1$.
\end{itemize}

In either case, agreement is achieved. This completes the proof.
\end{proof}



\subsection{Multi-valued Consensus}
\label{s:mvc}

It is easy to see that Algorithm Min-Max does not work correctly when the input is not binary, since the source(s) of some reduced graph may not have either maximum or minimum input value(s), and thus, the rest of the nodes cannot learn the value(s) of the source(s) in either Min or Max Phase. This section considers the consensus problem with input being in the range $[0, K]$, where $K \geq 1$. 

We present Algorithm MVC (Multi-Valued Consensus). 
It consists of two loops: The OUTER-LOOP iterates over all possible inputs, and the INNER-LOOP is essentially Algorithm Min-Max with an extra step to update the tentative state. In Algorithm MVC, each node $i$ keeps track of two types of variables:

\begin{itemize}
\item $t_i$: This variable is the tentative state at each node. It is guaranteed that at any point of time, $t_i$ equals an input at some node. Moreover, if node $i$ enters OUTER-LOOP iteration $l$, then $t_i$ is set to be some input value that has been seen by node $i$ and is at least $l$.

\item $v_i$: This {\em binary} variable acts as several roles. It first represents whether or not $t_i = l$ at the beginning of each OUTER-LOOP iteration $l$ (STEP I of the OUTER-LOOP). Then, at the end of STEP II of the OUTER-LOOP, $v_i$ becomes the output of Algorithm Min-Max (INNER-LOOP). Thus, at the beginning of STEP III of the OUTER-LOOP, nodes will have the same $v_i$'s, which allows nodes to reach an agreement on whether to proceed to next OUTER-LOOP iteration.


\end{itemize}

Now, we describe the structure of Algorithm MVC. In each OUTER-LOOP iteration $l \in [0, K]$, nodes try to learn whether some node $i$ has the tentative state $t_i = l$ at the beginning of the iteration. First, $v_i$ acts as a local observation at node $i$, i.e., $v_i$ is set to $0$ if and only if $t_i = l$ (STEP I of the OUTER-LOOP). Then, at STEP II of the OUTER-LOOP, nodes use Algorithm Min-Max (INNER-LOOP) to reach agreement on the observations ($v_i$'s). There are two possible outcomes at the end of the STEP II of the OUTER-LOOP:

\begin{itemize}
\item $v_i = 0$:

This case implies that nodes learn that some node $i$ has $t_i = l$ at the beginning of the OUTER-LOOP iteration, and they know that all the other nodes that have not crashed also learn the same information. Thus, nodes will exit the OUTER-LOOP with outputs $l$ (STEP III of the OUTER-LOOP).

\item $v_i = 1$:

In this case, nodes will proceed to the next OUTER-LOOP iteration.\footnote{Note that this case does not mean that {\em no} node has $t_i = l$. However, in this case, nodes cannot be sure that all nodes that have not crashed also have learned that some node $i$ has $t_i = l$ at the beginning of the OUTER-LOOP iteration. Thus, nodes have to proceed to the next OUTER-LOOP iteration.} Moreover, nodes are guaranteed to set their tentative state ($t_i$'s) to some value strictly greater than $l$ when completing the INNER-LOOP. At step 4 of each INNER-LOOP phase, nodes update $t_i$'s to the minimum value that is strictly greater than $l$ and is received in that INNER-LOOP phase. Later, we will show that if at any point of time, node $i$ changes $v_i$ from $0$ to $1$, then $t_i$ will also be updated to some value strictly greater than $l$. Thus, if nodes enter the OUTER-LOOP iteration $l+1$, then no node will ever have tentative state $\leq l$. If at the end of OUTER-LOOP $K$, nodes do not exit the loop, i.e., the code {\em Exit OUTER-LOOP} is never executed, then all the fault-free nodes will terminate with output $K$.
\end{itemize}
Note that due to the agreement property of Algorithm Min-Max, either nodes will exit OUTER-LOOP at the same iteration, or nodes will terminate with output $K$.


~

\hrule

\vspace*{2pt}

\noindent {\bf Algorithm MVC}

\vspace*{4pt}

\hrule

\vspace*{4pt}

\begin{itemize}

  \item $t_i[0] :=$ input at node $i$

  \item \textbf{OUTER-LOOP $l := 0$ to $K$}:

  \begin{itemize}  
  \item \textbf{STEP I:}
  If $t_i[l] == l$, then $v_i[l] := 0$; otherwise, $v_i[l] := 1$
 
  \item \textbf{STEP II:} \textit{INNER-LOOP $p := 1$ to $2f+2$}:

    ~~Repeat the following steps $d$ times:            
    
            \begin{enumerate}
            \item Broadcast the tuple $(v_i[l], t_i[l])$
            \item Receive the broadcast tuples from incoming neighbors and node $i$ itself. Denote by $B_i$ the set of tuples received in this step.
            \item 
            If $p~mod~2 = 0$, then \hfill $\backslash\backslash$ \textbf{Min-Phase} 
            
            ~~~~$v_i[l] := \min \{a~|~(a, *) \in B_i\}$
            
            Else, \hfill $\backslash\backslash$ \textbf{Max-Phase} 

            ~~~~$v_i[l] := \max \{a~|~(a, *) \in B_i\}$
           
            \item If $~|\min \{b~|~(*, b) \in B_i,~b > l\}| > 0$, then
                           
            ~~~~$t_i[l] := \min \{b~|~(*, b) \in B_i,~b > l\}$
            
            
            \end{enumerate}

   \item \textbf{STEP III:} If $v_i[l] == 0$, then 
   
   ~~~~ Exit OUTER-LOOP

   \end{itemize}

   \item Output $l$
   
   \end{itemize}

\hrule

~

\begin{theorem}
\label{thm:mvc}
Algorithm MVC is correct in all the graphs that satisfy $f$ CT node connectivity.
\end{theorem}

The proof is presented in Appendix \ref{a:mvc}.

\section{Iterative Algorithms}
\label{s:iterative}

Observe that Algorithm Min-Max does not utilize any topology information, since it does not require node identifiers (ID), and the usage of the fault-tolerant diameter $d$ can be replaced by the number of nodes $n$. That is, assuming the knowledge of $n$ and $f$, Algorithm Min-Max works in {\em anonymous} systems \cite{anonymous} and {\em anonymous} networks \cite{anonymous_networks}, where nodes do not have IDs. In anonymous systems, we define a family of iterative algorithms -- {\em Fixed Iterative Algorithm} -- those iterative algorithms using {\em fixed} transition functions. This section assumes synchronous systems, as well.

\paragraph{Iterative Algorithms}

We first describe the structure of the iterative algorithms of interest. Each node $i$ maintains state $v_i$, with $v_i[t]$ denoting the state of node $i$ at the {\em end}\, of the $t$-th iteration of the algorithm. Initial state of node $i$,
$v_i[0]$, is equal to the initial {\em input}\, provided to node $i$.
At the {\em start} of the $t$-th iteration ($t>0$), the state of
node $i$ is $v_i[t-1]$.
The iterative algorithms of interest will require each node $i$
to perform the following three steps in iteration $t$, where $t>0$.

\begin{enumerate}
\item {\em Transmit step:} Transmit current state, namely $v_i[t-1]$, on all outgoing edges.

\item {\em Receive step:} Receive values on all incoming edges. 
Denote by $r_i[t]$ the union of $i$'s value and the values received by node $i$ from its
neighbors.

\item {\em Update step:} Node $i$ updates its state using a transition function $Z_i$ as
follows. $Z_i$ is a part of the specification of the algorithm, and takes
as input the vector $r_i[t]$.
\begin{eqnarray}
v_i[t] & = &  Z_i ~(r_i[t], t)
\label{eq:Z_i}
\end{eqnarray}
\end{enumerate}

\paragraph{Fixed Iterative Algorithms}

\begin{definition} {\bf (Fixed Transition Function)}
A transition function $Z_i$ for node $i$ is said to be {\em fixed} if for all iteration $t \geq 0$ and all $i \in \sv$, $Z_i(R_i[t], t) = Z^*(R_i[t])$. In other words, the transition function does not change over time, and every node uses the same transition function.
\end{definition}


For iterative algorithms that use fixed transition function, we present the following result.

\begin{theorem}
\label{thm:fixed}
In general, it is impossible to solve consensus using {\em fixed iterative algorithms} in anonymous systems and networks.
\end{theorem}

The proof is presented in Appendix \ref{a:fix}.

\section{Asynchronous Systems}
\label{s:async}

This section considers asynchronous systems, where each node proceeds in different speed and the messages may be arbitrarily delayed. For simplicity, we assume the channels are reliable.

\paragraph{Approximate Consensus}

\cite{impossible_proof_lynch} showed that it is impossible to achieve exact consensus in asynchronous systems with a single crash fault. Therefore, we are interested in approximate consensus algorithms. The algorithms must achieve the following three properties:

\begin{itemize}
\item {\bf $\epsilon$-agreement}: the difference between outputs at any two fault-free nodes is bounded by $\epsilon$.

\item {\bf Validity}: the output at any fault-free node is in the {\em convex hull} of all the inputs.

\item {\bf Termination}: every fault-free node decides on an output in a finite-amount of time.
\end{itemize} 

\subsection{Necessity}

To facilitate the discussion, we first introduce an useful definition.

\begin{definition}
Given a graph $G(\sv, \se)$ and a node-partition $A, B$ of $\sv$, $A$ is said to {\bf propagate} to $B$ if (i) $B$ is not empty; and (ii) there exist at least $f+1$ distinct nodes in $A$ which have outgoing links to some node in $B$, i.e., $|\{i~|~i \in A,~~N_i^+ \cap B \neq \emptyset\}| \geq f+1$.
\end{definition}

We will denote the fact that set $A$ propagates to set $B$ by the notation of $A \Zightarrow B$. When it is not true that $A \Zightarrow B$, we will denote that fact by $A \not\Zightarrow B$.

\begin{theorem}
\label{thm:nc_async}
Suppose that an asynchronous approximate consensus algorithm exists for $G(\sv, \se)$. Then for any node partition $L, C, R$ of $\sv$, where $L$ and $R$ are both non-empty, either $L \cup C \Zightarrow R$ or $C \cup R \Zightarrow L$.
\end{theorem}

\begin{proof}
The proof is by contradiction. Suppose that there exists a correct approximate consensus algorithm, and $G(\sv, \se)$ does not satisfy the condition. That is, there exists a node partition $L, C, R$ such that $L$ and $R$ are not empty, and $L \cup C \not\Zightarrow R$ and $C \cup R \not\Zightarrow L$. Let $O(L)$ denote the set of nodes in $C \cup R$ that have outgoing links to some nodes in $L$, i.e., $\{i~|~i \in C \cup R,~~N_i^+ \cap L \neq \emptyset\}$. Similarly, define $O(R) = \{j~|~j \in L \cup C,~~N_j^+ \cap R \neq \emptyset\}$. By assumption, $|O(L)| \leq f$ and $|O(R)| \leq f$.

Consider the scenario where (i) each node in $L$ has input $0$; (ii) each node in $R$ has input $2\epsilon$; (iii) nodes in $C$ (if non-empty) have arbitrary inputs in $[0, 2\epsilon]$; (iv) no node crashes; and (v) the message delay from $O(L)$ to $L$ and from $O(R)$ to $R$ is arbitrarily large compared to all the other traffic. Consider nodes in $L$. From their perspectives, it is possible that all nodes in $O(L)$ have crashed. This is due to the following observations:

\begin{itemize}
\item The only nodes in $C \cup R$ that have outgoing links to $L$ are nodes in $O(L)$. Thus, nodes in $L$ are not able to learn whether nodes in $O(L)$ are alive or not from nodes in $(C \cup R) - O(L)$.

\item The message delay from $O(L)$ is arbitrarily large.

\item The size of $|O(L)| \leq f$.
\end{itemize} 

Therefore, nodes in $L$ cannot wait for any message from nodes in $O(L)$ to decide the outputs. Similarly, nodes in $R$ cannot wait for any message from nodes in $O(R)$ to decide the outputs. Consequently, to satisfy the validity property, the output at each node in $L$ has to be $0$, since $0$ is the input of all the nodes in $L$. Similarly, all nodes in $R$ have to output $2 \epsilon$. Thus, $\epsilon$-agreement property is violated, since $\epsilon < 2 \epsilon$. This is a contradiction.
\end{proof}

\subsection{Sufficiency}

We prove that the condition in Theorem \ref{thm:nc_async} is also sufficient by proposing an asynchronous approximate consensus algorithm -- Algorithm WA (Wait-and-Average). The algorithm assumes the knowledge of global topology at each node, and the algorithm proceeds in phases. In each phase, nodes flood messages containing their current value, ID (identifier), and a phase index. Each node $i$ waits until it has received {\em enough} values from other nodes. Then, node $i$ updates its value to be the \textit{average} of all the values received in this phase, and then proceeds to the next phase. When node $i$ has finished $p_{end}$ phases, it outputs its current state. $p_{end}$ is some sufficiently large integer.

Now, we discuss how many values received by a node is considered {\em enough}. Let $heard_i[p]$ be the set of nodes from which node $i$ has received values {\em during} phase $p$. Each node $i$ proceeds to perform the averaging operation if the following condition holds.

~

\noindent {\bf Condition {\em WAIT}}: ~ Denote by $reach_i(F)$ the set of nodes that have paths to node $i$ in the reduced graph $G_F$. Then, Condition {\em WAIT} is satisfied if there exists a set of nodes $F_i \subseteq \sv-\{i\}$ and $|F_i| \leq f$ such that $reach_i(F_i) \subseteq heard_i[p]$.\footnote{$reach_i(F_i)$ may be different in each phase, since it depends on the delay pattern. For simplicity, we ignore the phase index $p$ in the notation.}

~

Now, we present the algorithm below.

~

\hrule

\vspace*{2pt}

\noindent {\bf Algorithm WA}

\vspace*{4pt}

\hrule

\vspace*{4pt}

$p_{end}$ is some sufficiently large integer.

  \begin{itemize}
  \item For each node $i$, set $v_i[0]$ to the input at node $i$
  
  \item For Phase $p = 1$ to $p_{end}$:

        \begin{itemize}
        \item On entering phase $p\geq 1$:

            ~~~~$R_i[p] = \{v_i[p-1]\}$\vspace{1mm}
            
            ~~~~$heard_i[p] = \{i\}$\vspace{1mm}            
        
            ~~~~Send message $( v_i[p-1], i, p)$ to all the outgoing neighbors
        
            ~
        
        \item When message $(h, j, p)$ is received for the first time:
        
            ~~~~$R_i[p] = R_i[p] \cup \{h\}$\vspace{1mm}
            
            ~~~~$heard_i[p] = heard_i[p] \cup \{j\}$\vspace{1mm}            
        
            ~~~~Send message $(h, j, p)$ to all the outgoing neighbors\vspace{1mm}
            
            ~~~~if Condition {\em WAIT} holds:\vspace{1mm}
        
            ~~~~~~~~$v_i[p] = \frac{\sum_{v \in R_i[p]} v}{|R_i[p]|}$\vspace{1mm}
            
        \end{itemize}

   \item Output $v_i$
   \end{itemize}

\hrule

~

The following theorem shows the correctness of Algorithm WA. It also proves that the condition in Theorem \ref{thm:nc_async} is sufficient for approximate consensus in asynchronous systems.

\begin{theorem}
Algorithm WA is correct in all graphs that satisfy the condition in Theorem \ref{thm:nc_async}.
\end{theorem}

\begin{proofSketch}
Validity and termination properties are obvious. For $\epsilon$-agreement, we only present the key lemma here. The rest of the proof is standard, e.g., \cite{Tseng_podc14,AA_Dolev_1986,welch_book}.

For phase $p \geq 1$, consider two nodes $i, j$ that have successfully computed values $v_i[p]$ and $v_j[p]$, respectively, in phase $p$. That is, $i$ and $j$ have not crashed before computing $v[p]$'s. With a slight abuse of terminology, define $heard_i[p]$ as the set of nodes whose values are used by node $i$ to compute its state $v_i[p]$ in phase $p$. Define $heard_j[p]$ similarly.

\begin{lemma}
\label{lemma:alg_WA}
$heard_i[p] \cap heard_j[p] \neq \emptyset$.
\end{lemma}

\begin{proof}
By construction, there exist two sets $F_i$ and $F_j$ such that (i) $F_i \subseteq \sv - \{i\}$ and $|F_i| \leq f$; (ii) $F_j \subseteq \sv - \{j\}$ and $|F_j| \leq f$; (iii) $reach_i(F_i) \subseteq heard_i[p]$; and (iv) $reach_j(F_j) \subseteq heard_j[p]$. If $reach_i(F_i) \cap reach_j(F_j) \neq \emptyset$, then the proof is complete, since $reach_i(F_i) \subseteq heard_i[p]$ and $reach_j(F_j) \subseteq heard_j[p]$. Thus, $heard_i[p] \cap heard_j[p] \neq \emptyset$. Now, consider the case when $reach_i(F_i) \cap reach_j(F_j) =\emptyset$. We will derive a contradiction in this case.

We start with the following claim:

\begin{claim}
\label{claim:alg_WA}
In $G$, the only nodes that may have outgoing links to nodes in $reach_i(F_i)$ are nodes in $F_i$. Similarly, in $G$, the only nodes that may have outgoing links to nodes in $reach_j(F_j)$ are nodes in $F_j$.
\end{claim}

\begin{proof}
Recall that $reach_i(F_i)$ is defined as the set of nodes that have paths to node $i$ in the reduced graph $G_{F_i}$, and $reach_j(F_j)$ is defined similarly. Thus, $F_i \cap reach_i(F_i) = \emptyset$ and $F_j \cap reach_j(F_j) = \emptyset$. These two observations together with the definitions of $reach_i(F_i)$ and $reach_j(F_j)$ imply that there is no path from nodes in $\sv - reach_i(F_i) - F_i$ (if non-empty) to nodes in $reach_i(F_i)$ in $G_{F_i}$. Hence, the claim is proved.
\end{proof}

Let $L = reach_i(F_i)$, $R = reach_j(F_j)$ and $C = \sv - L - R$. Observe that since $reach_i(F_i) \cap reach_j(F_j) =\emptyset$, $L, C, R$ form a partition of $\sv$. Moreover, $i \in reach_i(F_i)$ and $j \in reach_j(F_j)$; hence, $L = reach_i(F_i)$ and $R = reach_j(F_j)$ are both non-empty. Then, let $O(L)$ be the nodes in $C \cup R$ that have outgoing links to some nodes in $L$ in $G$. Since $L = reach_i(F_i)$, the only nodes that may be in $O(L)$ are in $F_i$ due to Claim \ref{claim:alg_WA}. By assumption, $|F_i| \leq f$. Therefore, $C \cup R \not\Zightarrow L$. Similarly, we can argue that $L \cup C \not\Zightarrow R$. These two conditions violate the necessary condition, a contradiction. Thus, $reach_i(F_i) \cap reach_j(F_j) \neq \emptyset$, which implies $heard_i[p] \cap heard_j[p] \neq \emptyset$. This completes the proof.
\end{proof}

Let $M$ and $m$ denote the upper bound and the lower bound on the inputs, respectively. Then, by an analysis similar to \cite{Tseng_podc14,AA_Dolev_1986,welch_book}, Lemma \ref{lemma:alg_WA} can be used to show $\epsilon$-agreement when $p_{end}$ is sufficiently large (as a function of $n, f, M, m$). 

\end{proofSketch}

\section{Summary}

This paper addresses consensus problems in the presence of crash faults, where the underlying communication networks may be incomplete. We explore exact and approximate consensus algorithms in synchronous and asynchronous systems, respectively. We prove \textit{tight} conditions for the graphs to be able to solve these consensus problems.

\bibliography{paperlist}

\begin{thebibliography}{10}

\bibitem{welch_book}
H.~Attiya and J.~Welch.
\newblock {\em Distributed Computing: Fundamentals, Simulations, and Advanced
  Topics}.
\newblock Wiley Series on Parallel and Distributed Computing, 2004.

\bibitem{approximate_consensus_dynamic}
B.~Charron{-}Bost, M.~F{\"{u}}gger, and T.~Nowak.
\newblock Approximate consensus in highly dynamic networks.
\newblock {\em CoRR}, abs/1408.0620, 2014.

\bibitem{anonymous_networks}
C.~Delporte-Gallet, H.~Fauconnier, and A.~Tielmann.
\newblock Fault-tolerant consensus in unknown and anonymous networks.
\newblock In {\em ICDCS}, pages 368--375. IEEE Computer Society, 2009.

\bibitem{dolev_82_BG}
D.~Dolev.
\newblock The byzantine generals strike again.
\newblock {\em Journal of Algorithms}, 3(1):14–30, March 1982.

\bibitem{AA_Dolev_1986}
D.~Dolev, N.~A. Lynch, S.~S. Pinter, E.~W. Stark, and W.~E. Weihl.
\newblock Reaching approximate agreement in the presence of faults.
\newblock {\em J. ACM}, 33:499--516, May 1986.

\bibitem{impossible_proof_lynch}
M.~J. Fischer, N.~A. Lynch, and M.~Merritt.
\newblock Easy impossibility proofs for distributed consensus problems.
\newblock In {\em Proceedings of the fourth annual ACM symposium on Principles
  of distributed computing}, PODC '85, pages 59--70, New York, NY, USA, 1985.
  ACM.

\bibitem{anonymous}
R.~Guerraoui and E.~Ruppert.
\newblock What can be implemented anonymously?
\newblock In P.~Fraigniaud, editor, {\em DISC}, volume 3724 of {\em Lecture
  Notes in Computer Science}, pages 244--259. Springer, 2005.

\bibitem{psl_BG_1982}
L.~Lamport, R.~Shostak, and M.~Pease.
\newblock The byzantine generals problem.
\newblock {\em ACM Trans. on Programming Languages and Systems}, 1982.

\bibitem{AA_nancy}
N.~A. Lynch.
\newblock {\em Distributed Algorithms}.
\newblock Morgan Kaufmann, 1996.

\bibitem{lamport_agreement}
M.~Pease, R.~Shostak, and L.~Lamport.
\newblock Reaching agreement in the presence of faults.
\newblock {\em J. ACM}, 27(2):228--234, Apr. 1980.

\bibitem{Schmid_link}
U.~Schmid, B.~Weiss, and I.~Keidar.
\newblock Impossibility results and lower bounds for consensus under link
  failures.
\newblock {\em SIAM J. Comput.}, 38(5):1912--1951, Jan. 2009.

\bibitem{Tseng_Exact}
L.~Tseng and N.~H. Vaidya.
\newblock Exact byzantine consensus in directed graphs.
\newblock {\em CoRR}, abs/1208.5075, 2012.

\bibitem{Tseng_podc14}
L.~Tseng and N.~H. Vaidya.
\newblock Asynchronous convex hull consensus in the presence of crash faults.
\newblock In {\em Proceedings of the 2014 ACM Symposium on Principles of
  Distributed Computing}, PODC '14, pages 396--405, New York, NY, USA, 2014.
  ACM.

\bibitem{Graph_theory_west}
D.~B. West.
\newblock {\em Introduction To Graph Theory}.
\newblock Prentice Hall, 2001.

\end{thebibliography}


\begin{thebibliography}{10}

\end{thebibliography}

\appendix

\section{Proof of Theorem \ref{thm:mvc}}
\label{a:mvc}

\noindent {\bf Theorem \ref{thm:mvc}}
Algorithm MVC is correct in all the graphs that satisfy $f$ CT node connectivity.

~

\begin{proof}
The termination property is obvious. Now, we prove the two other properties. Suppose that the graph $G(\sv, \se)$ satisfies the condition stated in Theorem \ref{thm:nc}. Let $v_i^{end}[l]$ be the value $v_i[l]$ at node $i$ after the INNER-LOOP is completed in some OUTER-LOOP iteration $l$.

\begin{claim}
\label{claim:agree}
For all nodes $i, j$ that have not crashed in OUTER-LOOP iteration $l$, $v_i^{end}[l] = v_j^{end}[l]$.
\end{claim}

\begin{proof}
This is due to the correctness of Algorithm Min-Max, since if we ignore the code related to $t_i$'s, then the INNER-LOOP is essentially equal to Algorithm Min-Max.
\end{proof}

We will use Claim \ref{claim:agree} to prove the agreement property. In the proof below, we will say that a node {\em exits} an OUTER-LOOP iteration $l$ if it has $v_i^{end}[l] = 0$; otherwise, a node is said to {\em complete} the iteration $l$.

\begin{lemma}
\label{lemma:mvc-agree}
Algorithm MVC satisfies the agreement property in $G$.
\end{lemma}

\begin{proof}
By Claim \ref{claim:agree}, all the nodes that have not crashed will either exit the OUTER-LOOP in the same iteration $l$ or complete OUTER-LOOP iteration $K$. Thus, all the fault-free nodes will have the same output $l$.
\end{proof}

~

To prove the validity property, we first introduce some notations, and prove useful lemma and claims. Let $t_i^{begin}[l]$ be the value $t_i[l]$ at node $i$ at the \textit{beginning} of some OUTER-LOOP iteration $l$, and let $t_i^{end}[l]$ be the value $t_i[l]$ at node $i$ at the \textit{end} of OUTER-LOOP iteration $l$. Let $v_i^{begin}[l]$ be the value $v_i[l]$ at node $i$ after STEP I of the OUTER-LOOP iteration $l$. Thus, if $t_i^{begin}[l] = l$, then $v_i^{begin}[l] = 0$; otherwise, $v_i^{begin}[l] = 1$. 

\begin{lemma}
\label{lemma:ti>l}
In an OUTER-LOOP iteration $l~~(0 \leq l < K)$, for each node $i \in \sv$ that has not crashed, and has $v_i[l] = 1$, then $t_i[l] > l$.
\end{lemma}

\begin{proof}
The proof is by induction on OUTER-LOOP iterations.

\noindent {\em Induction Basis}: $l = 0$.

We first prove the following claim.

\begin{claim}
\label{claim:ti>l-1}
At any point of time, for each node $i$ that has not crashed, and has $v_i[0] = 1$, then $t_i[0] > 0$.
\end{claim}

\begin{proof}
First, we prove the following claim: each node $i$ will change $v_i[0]$ from $0$ to $1$ if and only if it receives $(1, x)$ from its incoming neighbor such that $x > 0$. 
The proof is by contradiction. Consider the first Max-Phase $p$ (of the INNER-LOOP) in which some node $i$ changes $v_i[0]$ from $0$ to $1$, because $i$ has received $(1, 0)$ from its incoming neighbors.
Then, consider a chain of nodes propagating the tuple $(1,0)$ from some node $s$ to node $i$ such that node $s$ has $v_s[0] = 1$ and $t_s[0] = 0$ at the beginning of the Max-Phase $p$. Note that by assumption of $p$, node $s$ has never received $(1,0)$ from other nodes before Max-Phase $p$. Moreover, node $s$ has also never received $(1, x)$ such that $x > 0$ from other nodes before Max-Phase $p$, since otherwise, $t_s[0]$ would be updated to $x$ at step 4 of the INNER-LOOP. These two observations imply that $v_s[0] = 1$ and $t_s[0] = 0$ before entering the INNER-LOOP, i.e., after line 1 of the OUTER-LOOP is executed. This is a contradiction.

Second, Claim \ref{claim:ti>l-1} follows directly from the claim above.



\end{proof}

Claim \ref{claim:ti>l-1} implies that the statement of Lemma \ref{lemma:ti>l} holds for the base case ($l = 0$).

~

\noindent {\em Induction Step}: Suppose that for all OUTER-LOOP iteration $l \geq r$, the statement of Lemma \ref{lemma:ti>l} holds. Consider the $(r+1)$-th OUTER-LOOP iteration. We can prove the following claim based on similar logic as in the base case and the induction hypothesis.

\begin{claim}
\label{claim:ti>l-2}
At any point of time, for each node $i$ that has not crashed and has $v_i[r+1] = 1$, then $t_i[r+1] > r+1$.
\end{claim}

This claim completes the proof of induction step. Thus, Lemma \ref{lemma:ti>l} is proved.
\end{proof}

\begin{claim}
\label{claim:ti}
At any point of time in an OUTER-LOOP iteration $l$, if node $i$ has not crashed, then  $t_i[l]$ equals an input at some node.
\end{claim}

\begin{proof}
This claim holds by construction, since all the $t$'s propagated are initially some node's input.
\end{proof}

\begin{claim}
\label{claim:output}
If any node $i$ exits OUTER-LOOP iteration $l$ and outputs $l$, then there must exist some node $j$ such that $t_j^{begin}[l] = l$.
\end{claim}

\begin{proof}
Suppose by way of contradiction that every node $j$ that has not crashed has $t_j^{begin}[l] \neq l$, and node $i$ exits iteration $l$. The first assumption implies that every node $j$ has $v_j^{begin}[l] = 1$. Due to the validity of Algorithm Min-Max, every node $j$ that has not crashed after completing INNER-LOOP has $v_j^{end}[l] = 1$. Therefore, no node will exit iteration $l$, a contradiction.
\end{proof}

Now, we are ready to prove the key lemma.

\begin{lemma}
\label{lemma:mvc-validity}
Algorithm MVC satisfies the validity property in $G$.
\end{lemma}

\begin{proof}
Consider two cases:

\begin{itemize}
\item Some node has input $K$:

In this case, suppose that all the fault-free nodes exit the OUTER-LOOP iteration $l \leq K$ and output $l$. Then, by Claims \ref{claim:ti} and \ref{claim:output}, the validity property holds. Suppose that no fault-free node exits the OUTER-LOOP, i.e., for all $i$ that has not crashed, $v_i^{end}[K] = 1$. In this case, the validity property still holds, since all the fault-free nodes will output $K$, and by assumption, some node has input $K$.

\item No node has input $K$:

Assume that all the nodes have input $\leq K'$, where $K' < K$. In this case, we show the following claim.

\begin{claim}
\label{claim:K'}
All the fault-free nodes will exit the OUTER-LOOP in some iteration $l \leq K'$.
\end{claim}
\begin{proof}
If fault-free nodes exit during some OUTER-LOOP iteration $l < K'$, then the proof is done. Suppose not. Then, in iteration $K'-1$, every node $i$ that has not crashed has $v_i^{end}[K'-1] = 1$. Consequently, by Lemma \ref{lemma:ti>l}, every node that has not crashed has $t_i^{end}[K'-1] > K' - 1$. This observation together with Claim \ref{claim:ti} and the assumption that the input is bounded by $K'$ imply that $t_i^{end}[K'-1] = K'$. Therefore, in the beginning of iteration $K'$, every node that has not crashed has $t_i^{begin}[K'] = K'$ and $v_i^{begin}[K'] = 0$. Then, due to the validity property of Algorithm Min-Max, every node $i$ that has not crashed has $v_i^{end}[K'] = 0$. Therefore, every fault-free node will exit the OUTER-LOOP in iteration $K'$.
\end{proof}

Claims \ref{claim:ti}, \ref{claim:output} and \ref{claim:K'} together prove the validity property.
\end{itemize}
\end{proof}

Lemmas \ref{lemma:mvc-agree} and \ref{lemma:mvc-validity} prove Theorem \ref{thm:mvc}.
\end{proof}

\section{Proof of Theorem \ref{thm:fixed}}
\label{a:fix}

\noindent {\bf Theorem \ref{thm:fixed}}
In general, it is impossible to solve consensus using {\em fixed iterative algorithms} in anonymous systems and networks.

~

\begin{proof}
We prove the theorem by showing a counter example. We present a directed graph that satisfies $f$ CT node connectivity, and show that no fixed transition function solves consensus. 

Consider a directed graph $G$ consisting of three parts: (i) a clique of size $f+1$, (ii) a source node $s$ that has an outgoing edge to every node in the clique, and (iii) a leaf node $l$ that has an incoming edge from every node in the clique. Note that there is no incoming edge to $s$, and no outgoing edge from $l$. Moreover, edge $(s,l)$ is not an edge in $G$. Obviously, the graph satisfies $f$ CT node connectivity, since (i) if $s \in F$, then at least one node in the clique is the source in the reduced graph $G_F$; (ii) if $s \not\in F$, then $s$ is the source in $G_F$.

Suppose that each node uses the transition function $Z$. First, we look at how $Z$ maps to a value when a node receives exactly $f+1$ values. Recall that $R_i[t]$ denotes the set of values received by $i$ at iteration $t$. It is clear that if $R_i[t]$ contains all $0$'s or all $1$'s, then $Z(R_i[t])$ should map to $0$ or $1$, respectively; otherwise, either validity or agreement property is violated. This implies the following claim:

\begin{claim}
\label{claim:0to1}
There must exist a pair of set of $2f+1$ values $R^0$ and $R^1$ such that (i) $|R^0| = |R^1| = f+1$; (ii) there is exactly one more $1$ in $R^1$ than in $R^0$, i.e., suppose $R^0$ contains $a$ $0$'s and $(f+1-a)$ $1$'s, then $R^1$ contains $(a-1)$ $0$'s and $(f+2-a)$ $1$'s; and (iii) $Z(R^0) = 0$ and $Z(R^1) = 1$. 
\end{claim}

Denote by $R$ the set of $f$ $0$'s and one $1$, and $R'$ the set of $f$ $1$'s and one $0$. Claim \ref{claim:0to1} implies that there are three possible cases. In all the cases below, we consider an execution of the algorithm where (i) a single node in the clique crashes before the algorithm starts; and (ii) no other node crashes throughout the execution. The inputs at each node is described in each case below.

\begin{itemize}

\item $Z(R) = 0$:

Consider the case when the source $s$ has input $1$, and all the other nodes have input $0$. Since the source node does not receive any value, its state can only be $1$ throughout the execution. For each node in the clique in each iteration $t \geq 0$, it receives $f$ $0$'s and one $1$, and thus, state at each node in the clique can only be $0$, since $Z(R) = 0$. Thus, the agreement property is violated. 

\item $Z(R) = 1$ and $Z(R') = 0$:

Consider the case when the leaf $l$ has input $0$, and all the other nodes have input $1$. Since the source node does not receive any value, its state can only be $1$ throughout the execution. For each node in the clique in each iteration $t \geq 0$, it receives $f+1$ $1$'s, and thus, state at each node in the clique can only be $1$. As a result, in each iteration $t$, the leaf node $l$ receives $f$ $1$'s and one $0$, and thus the state at node $l$ is $0$ in each iteration, since $Z(R') = 0$. Thus, the agreement property is violated. 

\item $Z(R) = 1$ and $Z(R') = 1$:

Consider the case when the source $s$ has input $0$, and all the other nodes have input $1$. Since the source node does not receive any value, its state can only be $0$ throughout the execution. For each node in the clique in each iteration $t \geq 0$, it receives $f$ $1$'s and one $0$, and thus, state at each node in the clique can only be $1$, since $Z(R') = 1$. Thus, the agreement property is violated. 

\end{itemize}

\end{proof}

\comment{++++++

+++++}

\end{document}